\documentclass[11pt]{article}
\usepackage[utf8]{inputenc}
\usepackage{setspace}
\usepackage[margin=1in]{geometry}
\usepackage{apacite}
\usepackage{blindtext}
\usepackage{amsmath,bm}
\usepackage{color}
\usepackage{rotating}
\usepackage{caption}
\usepackage{subcaption}
\usepackage{grffile}
\usepackage{graphicx}
\usepackage{lscape}
\usepackage{blindtext}
\usepackage{amssymb}
\usepackage{varioref} 
\usepackage{textcomp}
\usepackage{booktabs}
\usepackage{multirow}
\usepackage{pdflscape}
\usepackage{graphicx}
\usepackage{relsize}
\usepackage{tablefootnote}
\usepackage{wrapfig}
\usepackage{natbib}
\usepackage{listings}
\usepackage{longtable}
\usepackage{url}
\usepackage{tabularx}
\usepackage{float}
\usepackage{epsfig}
\usepackage[titletoc]{appendix}
\usepackage{arydshln}
\usepackage{footnote}
\usepackage{secdot}
\usepackage{bigints}
\usepackage{amsmath}
\usepackage{amssymb}
\usepackage{amsthm}
\usepackage[english]{babel}
\usepackage{xr}
\usepackage{adjustbox}
\usepackage{graphicx} 

\usepackage{multicol}
\usepackage{enumitem}
\usepackage{bigstrut}
\usepackage{rotating}
\usepackage{xr}
\usepackage{sectsty}
\subsubsectionfont{\normalfont\itshape}
\usepackage{algorithm}
\usepackage{algpseudocode}
\newtheorem{theorem}{Theorem}[section]

\newtheorem{result}{Result}[section]

\begin{document}

\title{Optimal Accelerated Life Testing Sampling Plan Design with Piecewise Linear Function based Modeling of Lifetime Characteristics}

\author{\bf{Sandip Barui${}^{1,\footnote{Corresponding author. E-mail address: sandipbarui@isical.ac.in \newline \indent}}$ and Shovan Chowdhury${}^{2}$}  \\\\
${}^{1}$Interdisciplinary Statistical Research Unit, \\ Indian Statistical Institute, Kolkata, 700108, West Bengal, India\\ [0.5ex]
${}^{2}$Quantitative Methods and Operations Management Area,\\ 
Indian Institute of Management Kozhikode, Kozhikode, 673570, Kerala, India.\\
}
\date{}

\maketitle

\begin{abstract}
Researchers have widely used accelerated life tests to determine an optimal inspection plan for lot acceptance. All such plans are proposed by assuming a known relationship between the lifetime characteristic(s) and the accelerating stress factor(s) under a parametric framework of the product lifetime distribution. As the true relationship is rarely known in practical scenarios, the assumption itself may produce biased estimates that may lead to an inefficient sampling plan. To this endeavor, an optimal accelerating life test plan is designed under a Type-I censoring scheme with a generalized link structure similar to a spline regression, to capture the nonlinear relationship between the lifetime characteristics and the stress levels. Product lifetime is assumed to follow Weibull distribution with non-identical scale and shape parameters linked with the stress factor through a piecewise linear function. The elements of the Fisher information matrix are computed in detail to formulate the acceptability criterion for the conforming lots. The decision variables of the sampling plan including sample size, stress factors, and others are determined using a constrained aggregated cost minimization approach and variance minimization approach. A simulated case study demonstrates that the nonlinear link-based piecewise linear approximation model outperforms the linear link-based model.
\end{abstract}


\begin{table}[!htbp]
\centering
\caption{Table for some essential$^*$ abbreviations}\label{Table00}
\begin{tabular}{l l }
\hline
 AIC & Akaike information criterion \\
 LTSP & life testing sampling plan \\
 ALT & Accelerated life test \\
 ALTSP & Accelerated life testing sampling plan\\
 SD & Standard deviation \\
 COV & Covariance \\
 CDF & Cumulative distribution function \\
 PDF & Probability density function\\
 EV & Extreme valued \\
 SEV & Standard extreme valued \\
 ML & Maximum likelihood \\
 OC & Operating characteristic \\
 PL & Piecewise linear \\
 PLA & Piecewise linear approximation\\
 PLF & Piecewise linear function\\
 SSE & Sum of squared errors\\ 
\hline
\end{tabular}  
\\
\footnotesize{*Some non-essential abbreviations are omitted for the sake of brevity. These are defined within the article.}
\end{table}

\begin{table}[!htbp]
\centering
\caption{Table for some essential$^*$ notations and symbols}\label{Table0}
\resizebox{\columnwidth}{!}{
\begin{tabular}{l l }
\hline
$N, n, m$ & Lot size, sample size and number of stress levels respectively\\
$s_0, \dots, s_m$ & Accelerating levels of the stress factor\\
$\xi_0, \dots, \xi_m$ & Standardized stress levels\\
$n_i, \pi_i$ & Sample size and proportion of allocation for the stress level $s_i$\\
$X_{ij}$ & Random lifetime of the $j$-th item under the stress level $s_i$\\ 
$T_{ij}$ & Random log-lifetime of the $j$-th item under the stress level $s_i$\\ 
$f_i(.), F_i(.)$ & pdf and cdf of Weibull lifetime distribution respectively\\ 
                & under the stress level $s_i$ \\
$\alpha_i, \lambda_i$ & Shape and scale parameters of Weibull distribution respectively\\
 & under the stress level $s_i$\\
$g_i(.), G_i(.),\bar{G}_i(.), h_i(.)$ & pdf, cdf, reliability and hazard functions of EV distribution respectively \\
& under the stress level $s_i$\\
$\mu_i, \sigma_i$ & Location and scale parameters of EV distribution respectively\\
& under the stress level $s_i$\\
$Q_1, Q_2$ & No. of cut-points for the link functions of $\mu_i$ and $\sigma_i$ respectively\\
$\Xi_{\mu}=\{\xi_{\mu, 0}, \dots, \xi_{\mu, Q_1}\}$& Set of cut-points to link $\mu_i$\\
$\Xi_{\sigma}=\{\xi_{\sigma, 0}, \dots, \xi_{\sigma, Q_2}\}$& Set of cut-points to link $\sigma_i$\\
$\bm \gamma_{\mu}=(\gamma_{\mu, 0}, \dots, \gamma_{\mu, Q_1})^{\intercal}$ & Parameter vector associated with linking $\mu_i$\\
$\bm \gamma_{\sigma}=(\gamma_{\sigma, 0}, \dots, \gamma_{\sigma, Q_2})^{\intercal}$ & Parameter vector associated with linking $\sigma_i$\\
$L_i$ & Likelihood function under the stress level $s_i$\\
$\hat l$ & Maximized value of the log-likelihood function \\
$\tau_0$ & Censoring time following Type-I time (right) censoring scheme\\  
$\bm \theta=(\bm \gamma_{\mu}^{\intercal}, \bm \gamma_{\sigma}^{\intercal})^{\intercal}$ & Full parameter vector\\
$\hat{\psi}$ & Maximum likelihood estimate of a parameter $\psi$\\
$F(\bm \theta), F^{-1}(\bm \theta)$ & Fisher information and covariance matrix respectively\\
$\mathbb{E}, \mathbb{V}$ & Expectation and variance \\
$k$ & Acceptability constant \\
$W$ & Statistic required to specify the acceptance criterion\\
$\sigma^2_{\mu_0}, \sigma^2_{\sigma_0}$ & Variances of $\widehat{\mu_0}$ and $\widehat{\sigma_0}$ repectively\\
$p_{nc}$ & Proportion of defectives or non-conforming items \\
${\alpha}, {\beta}$ & Producer’s risk and consumer’s risk respectively\\
$p_{\alpha}, p_{\beta}$ & Fraction of nonconforming items to be accepted \\
& with probability $1-\alpha$ and $\beta$ respectively\\
$z_{\alpha}$, $z_{1-\beta}$ & $\alpha$-th and $(1-\beta)$-th quantiles of \\
& standard normal distribution respectively\\
$u_{p_{\alpha}}$, $u_{p_{\beta}}$ & $p_{\alpha}$-th and $p_{\beta}$-th quantiles of \\
& standard EV distribution respectively\\
$c_a$ & Cost due to free replacement warranty\\
$\omega_1, \omega_2$ & Time since purchase of a product till which the free  replacement\\ 
& warranty is valid and time since expiry of free replacement \\
& warranty till which pro-rata warranty is valid respectively\\
$c_r, c_t, c^*$ & Unit cost of rejection, per unit time consumption \\
&  cost and unit cost of inspection respectively\\
$C_T$, $C_{\min}$ & Objective aggregate cost function and minimized value respectively \\
$\mathcal{V}_Q$, $V_{\min}$ & Objective variance of aggregate quantile of log-lifetime\\
&  function and minimized value respectively \\
\hline
\end{tabular}  
}\\
\footnotesize{*Some non-essential notations are omitted for the sake of brevity. These are defined within the article.}
\end{table}

\section{Introduction}\label{sec1}
Sampling inspection plays a significant role in determining the acceptance of lots in life testing experiments. The lifetime of durable products, not being an instantaneously measurable quality characteristic, requires to be inspected by an optimally designed life testing sampling plan (LTSP), also known as reliability sampling plan. The decision to accept or reject lots can be taken only when a requisite number of products fail within a smaller test duration. To ensure this, different censoring schemes and accelerating stress factors (viz., mechanical stress, temperature, voltage, humidity, and others) are employed which eventually make the experiment expensive. Hence, it is recommended for the test to be optimal with respect to certain technical criteria, e.g., total cost of the experiment, variance of an important lifetime characteristic, conditional value-at-risk, and so on. Based on these criteria, the optimal choice of the sample size, parameter settings for the censoring schemes, and levels of the accelerating factors are determined. Many life test sampling plans are proposed in the literature; see, e.g., \cite{chen2004optimal}, \cite{chen2004designing}, \cite{huang2008decisions}, \cite{tsai2008reliability}, \cite{hsieh2013risk}, 
 \cite{bhattacharya2015computation}, 
 \cite{fernandez2017economic}, \cite{chakrabarty2020optimum}, \cite{perez2023optimal}, 
 \cite{fernandez2023optimum} to name a few. \\

A life testing experiment is usually discontinued following a specific censoring scheme or combination of multiple censoring schemes. An appropriate scheme is chosen mostly based on applicability, necessity, and convenience. Commonly used censoring schemes include Type-I or time censoring (\citealp{kim2009reliability, li2018bayesian, seo2009design}), Type-II or item censoring (\citealp{chen2007bayesian}), generalized hybrid censoring (\citealp{chakrabarty2021optimum}), random censoring (\citealp{chen2004optimal}), progressive censoring (\citealp{fernandez2011design, lee2024construction}) and so on. Type-I censoring is one of the most popular censoring schemes where a future time point or time duration ($\tau_0 >0$) is pre-determined by the manufacturer or experimenter. Whenever a product unit fails in $[0,\tau_0)$, the failure time is observed. However, for a unit not failing in $[0, \tau_0)$, $\tau_0$ is considered as the censoring time for this unit. Therefore, the information obtained from the censored units is incomplete. \\              

Any LTSP designed under an accelerated set-up is popularly known as an accelerated life test (ALT) sampling plan. Most ALT studies assume that the relationship between a lifetime characteristic and accelerating stress factor is known under a parametric framework of the product lifetime distribution. A log-linear link function between accelerating stress and lifetime characteristics, e.g., mean or standard deviation (SD) or both, have been widely studied in literature; see, e.g., \cite{gouno2004optimal, pascual2005lognormal, 
bunea2006competing,  watkins2008constant, seo2009design, nelson2009accelerated, tseng2009optimal, wu2020optimal, ma2021optimal, chakrabarty2022economic,  zheng2023reliability}. An Arrhenius relationship between logarithmic mean and accelerating stress for a lognormal lifetime distribution (\citealp{nelson1971analysis}) and an inverse power law relationship between the scale parameter and accelerating stress for a Weibull lifetime distribution have also been studied (\citealp{nelson1980accelerated}). Further, there are studies where the location parameter related to lifetimes has been linked to the linear function of an accelerating stress factor (\citealp{xu2009optimum, alhadeed2005optimal}). While \cite{meeter1994optimum} investigated the effect of linear and log-linear relationships between the location and scale parameters with accelerating stress, \cite{seo2009design} examined log-linear relationships for both scale and shape parameters with accelerating stress for Weibull lifetime distribution. Both these articles developed accelerated life testing sampling plans (ALTSPs) for products having Weibull distributed lifetime with non-identical parameters. \\

In most practical scenarios, the true relationship between product lifetime or lifetime characteristics (e.g., mean, SD, skewness, etc.) is rarely known. In all aforementioned studies, assumptions of linearity, log-linearity, or other reasonable mathematical functions under a parametric framework are considered. Therefore, it is only logical to make assumptions about the relationship between accelerating stress and lifetime characteristics to be general, flexible, and inclusive which can capture the true relationship conveniently. To that end, \cite{si2017accelerated} used constrained penalized B-spline regression to study the relationship between lifetime characteristics and accelerating stress factors under the ALT framework. Along similar lines, in this paper, we propose a spline-like regression of lifetime characteristics on accelerating stress while considering product lifetimes to assume a Weibull distribution. We link the shape and scale parameter of the product lifetime distribution to a stress factor through a piecewise linear (PL) function, therefore, these lifetime characteristics can vary according to the stress levels. Here, note that the shape and scale parameter of the product lifetime distribution is related directly to the lifetime characteristics, such as mean, SD, etc. The piecewise linear function is marked by user-defined knots or cut-points over the range of values taken by the stress factor. In the case of standardized stress, the cut-points can be chosen in the interval $[0,1]$. These cut-points create a non-overlapping sequence of intervals, and two different intervals may represent linear functions with different intercepts and slopes. Hence, the overall PL function-based link is non-linear and customizable in nature. The number of cut-points to be chosen depends on the user or manufacturer. A higher number of cut-points provides greater non-linearity and flexibility to the linking structure. However, increasing the number of cut-points may lead to greater problem complexity and overfitting. Application of the PL function in survival analysis has been well explored; see \cite{Lar85,Lo93,Che01,Bal16b,Bala-Barui22,pal2022stochastic}.  We could also have considered a smoothing spline with higher order, however, that would have increased the number of model parameters substantially and made the optimization problem even more cumbersome and intractable \cite{hastie2009elements}. \\

As mentioned earlier, our proposed framework is more general and flexible to capture the true nature of the relationship between accelerating stress and lifetime characteristics, hence,  more accurate estimates can be obtained, and an efficient sampling plan can be designed. The work, that ensues, has the following aspects: \\
\begin{enumerate}
    \item The product lifetimes are assumed to follow non-identical but independent Weibull distribution. Weibull distribution is a member of the log-location scale family of distributions and is commonly used to analyze lifetime data. It has several interesting properties including a wide variety of shapes for density and hazard functions.
    \item Only one stress factor with multiple accelerated levels is considered for simplicity and tractability. The log-lifetime distribution characteristics (location and scale parameters as shown in the next section) are linked to the accelerating stress levels via the PL-based link functions.
    \item Items are subjected to a Type-I censoring scheme, due to its convenience. 
    \item A theoretically implicit Fisher information matrix is derived. The computation task of the elements of the matrix is carried out in great detail.
    \item The matrix elements are used to formulate and design a statistic. The asymptotic properties of the statistic are derived. Based on the properties, the acceptability criterion for a conforming product unit is suggested.    
    \item A constrained optimal ALTSP is proposed and designed based on aggregate cost minimization of products and overall variance minimization of the quantiles of product log-lifetime.  
\end{enumerate}

This article is organized in the following sequence. In Section \ref{sec2}, we extensively develop the model framework that includes defining the piecewise linear link function and deriving the likelihood function under the Type-I censoring scheme, and the corresponding Fisher information matrix. In Section \ref{sec3}, the operating characteristic (OC) function is defined where the development of an unbiased statistic is discussed, and derivation of the acceptability constant is carried out that ensures product quality standard. In Section \ref{sec4}, optimal ALTSPs are designed with respect to aggregate cost minimization of products under warranty and overall variance minimization of the quantiles of the log-lifetime. In Section \ref{sec5} results from the optimal ALTSPs under six different input parameter scenarios are discussed. In Section \ref{sec6} a simulated case illustrates a better fit of the PLA-based modelling of the lifetime characteristics when compared to linear link-based modelling. Finally, we summarize the findings in our paper in Section \ref{sec7} and provide a few future directions.

%%%%%%%%%%%%%%%%%%%%%%%%%%%%%%%%%%%%%%%%%%%%%%%%%%%%%%%%%%%%%%%%%%%%%%%%%%%
%%%%%%%%%%%%%%%%%%%%%%%%%%%%%%%%%%%%%%%%%%%%%%%%%%%%%%%%%%%%%%%%%%%%%%%%%%%
%%%%%%%%%%%%%%%%%%%%%%%%%%%%%%%%%%%%%%%%%%%%%%%%%%%%%%%%%%%%%%%%%%%%%%%%%%%
%%%%%%%%%%%%%%%%%%%%%%%%%%%%%%%%%%%%%%%%%%%%%%%%%%%%%%%%%%%%%%%%%%%%%%%%%%%
\section{Model Framework}\label{sec2}

\subsection{Product lifetime distribution}\label{sec2_1}
Throughout the paper, a single stress factor with $m$ increasing levels, namely, $s_1< \dots< s_m$ is considered with $s_0$ as the usage condition stress level. For $i=0, 1, \dots, m$, let $\xi_i=\frac{s_i-s_0}{s_m-s_0}$ be the standardized stress levels with $0 = \xi_0 < \xi_1 < \dots < \xi_m = 1$. Let $X_{ij}$ denote the actual lifetime of the $j$-th product unit subjected to the stress level $s_i$. For the purpose of modeling, the lifetime distribution of $X_{ij}$ can be considered parametric, semi-parametric, or distribution-free. However, due to its simplicity and greater flexibility with respect to the shape ($\alpha_i>0$) and scale ($\lambda_i>0$) parameters, a non-identical family of Weibull distributions is considered to characterize the product lifetimes (\citealp{kotz2004continuous}). Therefore, for $i=0, \dots, m$ and $x \ge 0$, 
\begin{equation}
f_{i}(x)=\left(\alpha_i \lambda_i^{\alpha_i}\right) x^{(\alpha_i-1)} e^{-\left(\lambda_i x\right)^{\alpha_i}} 
\end{equation}
and
\begin{equation}
F_{i}(x)= 1-e^{-\left(\lambda_i x\right)^{\alpha_i}} 
\end{equation}
are, respectively, the probability density function (pdf) and cumulative distribution function (cdf) of $X_{ij}$. Let $n_i$ be the number of units tested under stress level $s_i$. Clearly, $X_{ij}$ and $X_{ij'}$ are independently and identically distributed (iid) with  Weibull distribution having shape parameter $\alpha_i$ and scale parameter $\lambda_i$ for the units $j$ and $j'$ tested under the same stress level $s_i$. However, $X_{ij}$ and $X_{i'j'}$ are not identical (but independent) for the units $j$ and $j'$ tested under different stress levels $s_i \ne s_{i'}$. The additional flexibility consideration of the non-identical product lifetimes, where both shape and scale parameters are non-constant for different units under different stresses, ensures further model generalization. For more details on the practical implications of the non-constant shape parameter of Weibull distribution, one can refer to \cite{seo2009design}. The independence of lifetimes assumption has been kept in our model as it seems only logical that the lifetime of a unit is not influenced by the lifetime of another unit tested under the same or different stress levels. Therefore, for $i=0, \dots, m$, $-\infty<\mu_i<\infty$, $\sigma_i>0$ and $-\infty < t < \infty$, 
\begin{equation}\label{EQ2}
G_{i}(t)= 1-\exp\left[-e^{\left(\frac{t-\mu_i}{\sigma_i}\right)}\right] 
\end{equation}
is the cdf of $T_{ij}$ where $T_{ij}=\ln\{X_{ij}\}$ and 
\begin{equation} \label{SEV_pdf}
g_i(t)=\frac{d}{dt}G_i(t)={\sigma_i^{-1}}\exp\left[-e^{\left(\frac{t-\mu_i}{\sigma_i}\right)}\right] e^{\left(\frac{t-\mu_i}{\sigma_i}\right)}
\end{equation}
is the pdf of $T_{ij}$. Note that taking the transformation $X_{ij} \to T_{ij}$ converts the family of product lifetimes distribution to location-scale extreme-valued (EV) distribution where $\mu_i=-\ln\{\lambda_i\}$ and $\sigma_i=\alpha_i^{-1}$ (\citealp{bhattacharya2015computation}). The parameters $\mu_i$ (location) and $\sigma_i$ (scale) represent and involve product log-lifetime characteristics. For example, the mean and variance of log-lifetime for products tested under stress $s_i$ are given by $\mu_i-\sigma_i\gamma$ and $\sigma^2_i\left(\frac{\pi^2}{6}\right)$ respectively, where $\gamma \approx 0.5772$ is known as the Euler's constant. The hazard function of $T_{ij}$ is given by  \begin{equation}\label{haz}
h_i(t)=\frac{g_i(t)}{1-{G}_i(t)}=\sigma_i^{-1} e^{\left(\frac{t-\mu_i}{\sigma_i}\right)}    
\end{equation}
for $i=1, \dots, m$ and $j=1, \dots, n_i$. It is important to note that the log-lifetime characteristics $\mu$ and $\sigma$ can be related to stress ($s$), and hence, to standardized stress ($\xi$) linearly or non-linearly. This further implies that the lifetime (as opposed to log-lifetime) characteristics $\alpha$ and $\lambda$ can be related to stress ($s$) linearly or non-linearly. A detailed literature review and discussion on the non-linear relationship between lifetime characteristics and stress are given in Section \ref{sec1}.     

%%%%%%%%%%%%%%%%%%%%%%%%%%%%%%%%%%%%%%%%%%%%%%%%%%%%%%%%%%%%%%%
%%%%%%%%%%%%%%%%%%%%%%%%%%%%%%%%%%%%%%%%%%%%%%%%%%%%%%%%%%%%%%%
%%%%%%%%%%%%%%%%%%%%%%%%%%%%%%%%%%%%%%%%%%%%%%%%%%%%%%%%%%%%%%%
%%%%%%%%%%%%%%%%%%%%%%%%%%%%%%%%%%%%%%%%%%%%%%%%%%%%%%%%%%%%%%%
%%%%%%%%%%%%%%%%%%%%%%%%%%%%%%%%%%%%%%%%%%%%%%%%%%%%%%%%%%%%%%%

\subsection{Piecewise Linear Approximation of Lifetime Characteristics based on Accelerating Stress}\label{sec2_2}
%%%%%%%%%%%%%%%%%%%%%%%%%%%%%%%%%%%%%%%%%%%%%%%
As mentioned before, the log-lifetime parameters $\mu=\mu(\xi)$ and $\sigma=\sigma(\xi)$ are let to depend on the standardized stress $\xi$. Let  $Q_1$ and $Q_2$ denote 
numbers of cut-points (or knots) required to model $\mu(\xi)$ and $\sigma(\xi)$, respectively, where $Q_1$ and $Q_2$ are user-defined positive integers.  Also, let
$$\Xi_{\mu}=\{\xi_{\mu, 0}, \dots, \xi_{\mu, Q_1}; \xi_{\mu,0} < \dots < \xi_{\mu, Q_1}\}$$ and 
$$\Xi_{\sigma}=\{\xi_{\sigma, 0}, \dots, \xi_{\sigma, Q_2}; \xi_{\sigma,0} < \dots < \xi_{\sigma, Q_2}\}$$ be user-defined sets of cut-points. For $i=0, \dots, m$, the location parameter $\mu_i$ and the scale parameter $\sigma_i$ of the EV distribution of $T_{ij}=\ln\{X_{ij}\}$ are assumed to be linked with the standardized stress $\xi_i$ in the following manner: 
\begin{equation}\label{PLAdefmu}
    \mu_i=\mu(\xi_i; \bm \gamma_{\mu})=\sum_{q=1}^{Q_1}\left[\gamma_{\mu, q}+ \left\{\frac{\gamma_{\mu, q}-\gamma_{\mu, q-1}}{\xi_{\mu,q}-\xi_{\mu, q-1}}\right\} (\xi_i-\xi_{\mu, q}) \right] I_{[\xi_{\mu,q-1},\xi{\mu,q}]}(\xi_i)
\end{equation} and 

\begin{equation}\label{PLAdefsig}
   \ln \sigma_i=\ln \sigma(\xi_i; \bm \gamma_{\sigma})=\sum_{q=1}^{Q_2}\left[\gamma_{\sigma, q}+ \left\{\frac{\gamma_{\sigma, q}-\gamma_{\sigma, q-1}}{\xi_{\sigma,q}-\xi_{\sigma,q-1}}\right\} (\xi_i-\xi_{\sigma,q}) \right] I_{[\xi_{\sigma, q-1},\xi_{\sigma, q}]}(\xi_i)
\end{equation}
where \begin{equation*}
I_{[a,b]}(x)=\begin{cases}
1, & \text{if } x \in [a,b]\\
0, & \text{otherwise}
\end{cases}
\label{Indicator}    
\end{equation*}
for any real valued $b > a$. The models defined in the Equations (\ref{PLAdefmu}) and (\ref{PLAdefsig}) are known piecewise linear models. $\gamma_{\mu, q}; q=0, \dots, Q_1$ and $\gamma_{\sigma, q'}; q'=0, \dots, Q_2$ are the model parameters determining the intercept and slope of the link functions within the intervals $[\xi_{\mu, q-1}, \xi_{\mu, q}]$ and $[\xi_{\sigma, q-1}, \xi_{\sigma, q}]$, respectively. \\

For $i=0, \dots, m$,  the following assumptions are used:
$$\mu(\xi_{\mu,q})=\displaystyle{\lim_{\xi_i \to \xi_{\mu, q}} \mu(\xi_i; \bm \gamma_{\mu})}=\gamma_{\mu,q};~q=0, \dots, Q_1,$$
and
$$\sigma(\xi_{\sigma,q})=\displaystyle{\lim_{\xi_i \to \xi_{\sigma, q}} \sigma(\xi_i; \bm \gamma_{\sigma})}=\gamma_{\sigma,q};~q=0, \dots, Q_1$$
to ensure continuity at the cut-points. Further, note that the $\ln\{.\}$ transformation for $\sigma_i$ further ensures the non-negativity of the scale parameter. Therefore, $\mu(\xi)$ and $\sigma(\xi)$ are approximated by the piecewise linear functions, or estimated by piecewise linear approximations (PLAs), where the user needs to define the sets $\Xi_{\mu}$ and $\Xi_{\sigma}$.  For brevity, let $\bm \gamma_{\mu}=(\gamma_{\mu,0}, \dots, \gamma_{\mu,Q_1})^{\intercal}$ and $\bm \gamma_{\sigma}=(\gamma_{\sigma,0}, \dots, \gamma_{\sigma,Q_2})^{\intercal}$ where
 $\bm \gamma_{\mu}$ and $\bm \gamma_{\sigma}$ are to be estimated by implementing the maximum likelihood estimation method (discussed in Section \ref{sec2_3}). \\

 \subsubsection{Choice of Knots}\label{sec2_1_1}
 The choice of the cut-points, i.e., elements of the sets $\Xi_{\mu}$ and $\Xi_{\sigma}$ depends on the user. As a recommendation, we suggest that one should first decide on the number of linear functions by which the true relationship between any lifetime characteristic and stress is to be approximated (i.e., the values of $Q_1$ and $Q_2$ in our case). Generally, higher values of $Q_1$ (or $Q_2$) provide a better and more refined fit to the true relationship curve. However, increasing $Q_1$ (or $Q_2$) increases the number of parameters to be estimated resulting in the problem of overfitting. Hence, solving optimization problems either becomes computationally intensive and challenging, or a solution may not exist at all. Once $Q_1$ (or $Q_2$) is chosen, one way to choose the cut-points $\xi_{\mu, q}; q=0, \dots, Q_1$ (or $\xi_{\sigma, q'}; q'=0, \dots, Q_2$) is by considering them suitably from the set of equispaced quantiles of the stress levels $\xi_0, \dots, \xi_m$. That is, $$\Xi_{\mu}=\left\{\xi_{\mu, 0}=\min\{\xi_i\}=0,  \xi_{\mu, 1}=\tilde{q}_1 , \dots, \xi_{\mu, Q_1-1}=\tilde{q}_{Q_1-1}, \xi_{\mu, Q_1}=\max\{\xi_i\}=1\right\}$$
 where $\left\{\# j:  j \le  \tilde{q}_j\right\}=\left(\frac{j}{m+1}\right); j=1, \dots, Q_1-1$ 
 and $$\Xi_{\sigma}=\left\{\xi_{\sigma, 0}=\min\{\xi_i\}=0,  \xi_{\sigma, 1}=\tilde{q'}_1 , \dots, \xi_{\sigma, Q_2-1}=\tilde{q'}_{Q_2-1}, \xi_{\sigma, Q_2}=\max\{\xi_i\}=1\right\}$$
 where $\left\{\# j:  j \le  \tilde{q'}_j\right\}=\left(\frac{j}{m+1}\right); j=1, \dots, Q_2-1$. Note that, in the case taken above, both $Q_1 \le m+1$ and $Q_2 \le m+1$. Further, when data on stress levels are not present as in the case of considering the model as a part of an optimization problem where $\xi_i$s one requires the optimized value of $\xi_i$s, then one may take $\tilde{q}_j=\tilde{q'}_j=\left(\frac{j}{m+1}\right)$. Other methods of choosing the cut-points have been discussed in the literature. Interested readers may look into the works of  \cite{Bal16b}, \cite{Bala-Barui22}, \cite{pal2022stochastic} to gain further insights into PLA. \\

\subsubsection{Example of PLA-based Model}\label{sec2_1_2}
We consider the following example to motivate readers in favor of PLA-based models. For any standardized stress $\xi \in [0,1]$, we define the true relationship between a lifetime characteristic $\theta (\xi)$ and stress $\xi$ as  $$\theta (\xi)= \left(\frac{\xi-0.5}{0.2 \sqrt{2\pi}}\right) \exp\left\{-0.5\left(\frac{\xi-0.5}{0.2}\right)^2\right\}.$$ We fit the following functions of $\xi \in [0, 1]$ to approximate (estimate) $\theta(\xi)$: linear, logarithmic, inverse, combination of inverse and logarithmic, cubic polynomial, and piecewise linear functions (PLA1 - PLA3). The exact form and detail of each link (approximating) function are given in Table \ref{Table1}. PLA1 - PLA3 represent piecewise linear functions as defined in Equation (\ref{PLAdefmu}) or (\ref{PLAdefsig}) with 3 respective choices of cut-points, namely, $\Xi_{\theta}=(0.00,0.30,0.70,1.00)^{\intercal}$, $\Xi_{\theta}=(0.00,0.25,0.50,0.75,1.00)^{\intercal}$ and $\Xi_{\theta}=(0.00,0.30,0.50,0.70,1.00)^{\intercal}$. As defined by the user, PLA1 approximates the true function by three 3 piecewise linear functions whereas PLA2 and PLA3 use 4 piecewise linear functions to approximate the true function. For PLA1 and PLA3, cut-points are chosen visually in $[0, 1]$ by considering the points of inflection (change of slope) of the true function. For PLA2 cut-points are chosen as equispaced quantiles from a uniform distribution on $[0, 1]$, i.e. the points $\{0.00, 0.25, 0.50, 0.75, 1.00\}$. The function parameters are estimated by minimizing the sum of squared errors (SSEs), the minimized SSEs are displayed in Table \ref{Table1}.  It can be observed that PLA1 and PLA3 provide best fit to the true non-linear relationship in terms of SSE, even better than cubic polynomials. A simple graph demonstrating the advantage of fitting PLA over other standard relationships (linear, logarithmic, inverse, cubic, etc.) is displayed in Figure \ref{Fig1}, especially when the true relationship is non-linear and non-monotone.\\
% Add more on PLA here:  

%%%%%%%%%%%%%%%%%%%
%%%%%%%%%% Table 1  
\begin{table}[!htbp]
\caption{Comparison with respect to SSE among the fitted models}\label{Table1}
\centering
\begin{tabular}{c | c}
\hline
{\bf Fitted model } & {\bf SSE}  \\
\hline
Linear: $\theta(\xi)=\gamma_0+\gamma_1 \xi$ &  12.2316 \\

Logarithmic: $\theta(\xi)=\gamma_0+\gamma_1 \ln\{\xi\}$  & 18.3569  \\

Inverse: $\frac{1}{\theta(\xi)}=\gamma_0+\gamma_1 \xi$  &  12.7761 \\

Combination: $\frac{1}{\theta(\xi)}=\gamma_0+\gamma_1 \ln\{\xi\}$  &  17.3678 \\
Cubic: $\theta(\xi)=\gamma_0+\gamma_1 \xi+\gamma_2 \xi^2+\gamma_3\xi^3$ &  1.6923 \\

PLA1: $\Xi_{\theta}=(0.00,0.30,0.70,1.00)^{\intercal}$ & 0.5999  \\
%and $\bm \gamma_{\theta}=(1.20,1.00,1.50,1.30)^{\intercal}$

PLA2: $\Xi_{\theta}=(0.00,0.25,0.50,0.75,1.00)^{\intercal}$ &  2.8390 \\
%and $\bm \gamma_{\theta}=(1.20,1.05,1.20,1.50,1.30)^{\intercal}$

PLA3: $\Xi_{\theta}=(0.00,0.30,0.50,0.70,1.00)^{\intercal}$& 0.5702  \\
% and $\bm \gamma_{\theta}=(1.20,1.00,1.25,1.50,1.29)^{\intercal}$

\hline
\end{tabular}\\  
\footnotesize{SSE: Sum of squared errors; $\Xi_{\theta}:$Set of user-defined cut-points}
\end{table}

%%%%%%%%%%%%%%%%%%%%%%%%%%%%%%%%%%
\begin{figure}[!htbp]
\centering
\includegraphics[width=6.3in]{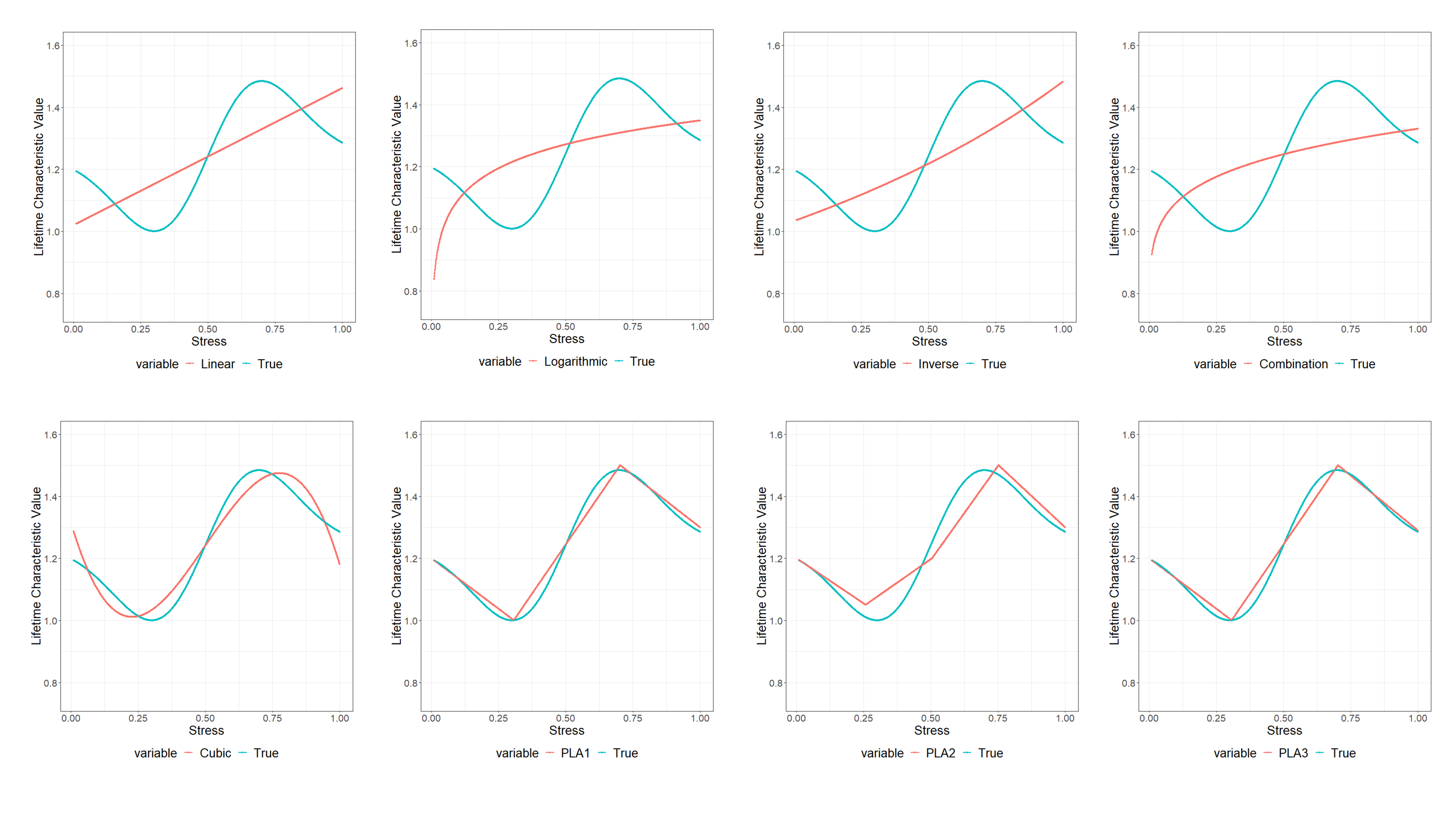}
\caption{A graphical illustration of the goodness of fit of the piecewise linear approximation model}
\label{Fig1}
\end{figure}
%%%%%%%%%%%%%%%%%%%%%%%%%%%%%%%%%%

\subsection{Likelihood Function under Type-I Censoring and Associated Fisher Information Matrix}\label{sec2_3}
For $i=0, \dots, m$, let $n_i$ denote the number of products tested under the stress level $s_i$ and $\tau_0$ denote the pre-determined fixed-time of censoring. For the Type-I censoring scheme (i.e., for the  fixed-time censored data),  the likelihood function under stress $s_i$ can be expressed by
\begin{equation}\label{likeli}
    L_i=L(\bm \theta; s_i) = L(\bm \theta; \xi_i)  \propto  \left(\bar{G}_i(\tau_0; \bm \theta)\right)^{n_i-d_i} {\prod}_{j \in \Delta^*_{i}} 
 g_i(t_{ij}; \bm \theta)
\end{equation}
 where $\bar{G}_i(t; \bm \theta) = 1-G_i(t; \bm \theta)$ is the reliability function with $G_i(t; \bm \theta)$ defined in Equation (\ref{EQ2}), $g_i(t; \bm \theta)$ defined in Equation (\ref{SEV_pdf}), $\Delta^*_{i}$ is the set of all units tested under $s_i$ that are not censored, $d_i=|\Delta^*_i|$ denote the observed number of unit failures and $\bm \theta = (\bm \gamma_{\mu}^\intercal, \bm \gamma_{\sigma}^\intercal)^\intercal$ is the involved model parameter vector. The quantity $\prod_{{i=1}}^m L_i$ or $\sum_{i=0}^m \ln L_i$ is maximized with respect to $\bm \theta$ over its parametric space to obtain the maximum likelihood (ML) estimator $\widehat{\bm \theta}=(\widehat{\bm \gamma}_{\mu}^{\intercal},\widehat{\bm \gamma}_{\sigma}^{\intercal})^{\intercal}$ of $\bm \theta=(\bm \gamma_{\mu}^\intercal, \bm \gamma_{\sigma}^\intercal)^\intercal$.

Let us define the following terms for $i=0, \dots, m$, $q, q^*=0, \dots, Q_1$ and $q', q^{+}=0, \dots, Q_2$:
\begin{flalign*}
\mu'_{i,q}=\frac{\partial \mu(\xi_i; \bm \gamma_{\mu})}{\partial \gamma_{\mu, q}}=\sum_{q=1}^{Q_1} \left(\frac{\xi_i-\xi_{\mu,q-1}}{\xi_{\mu,q}-\xi_{\mu,q-1}}\right) I_{[\xi_{\mu,q-1},\xi_{\mu,q}]}(\xi_i), 
\end{flalign*} 
\begin{flalign*}
\sigma'_{i,q'}=\frac{\partial \sigma(\xi_i; \bm \gamma_{\sigma})}{\partial \gamma_{\sigma, q'}}=\sigma(\xi_i,\bm \gamma_{\sigma})\sum_{q'=1}^{Q_2} \left(\frac{\xi_i-\xi_{\sigma,q'-1}}{\xi_{\sigma,q'}-\xi_{\sigma,q'-1}}\right) I_{[\xi_{\sigma,q'-1},\xi_{\sigma,q'}]}(\xi_i), \end{flalign*}
\begin{flalign*}
    A_{i}=A(\bm \gamma_{\mu}, \bm \gamma_{\sigma}; t, \xi_i)=\ln  h_i(t; \bm \theta)=\frac{t-\mu(\xi_i;\bm \gamma_{\mu})}{\sigma(\xi_i; \bm \gamma_{\sigma})} - \ln \sigma(\xi_i; \bm \gamma_{\sigma}),
\end{flalign*}
\begin{flalign*}
   A'_{i;\mu_q}=\frac{\partial A_{i}}{\partial \gamma_{\mu, q}}= -\frac{\mu'_{i,q}}{\sigma(\xi_i; \bm \gamma_{\sigma})},
   A'_{i;\sigma_{q'}}=\frac{\partial A_{i}}{\partial \gamma_{\sigma, q'}}= - \frac{(t-\mu(\xi_i;\bm \gamma_{\mu})) \sigma'_{i,q'}}{\sigma^2(\xi_i; \bm \gamma_{\sigma})} - \frac{\sigma'_{i,q'}}{\sigma(\xi_i; \bm \gamma_{\sigma})},
\end{flalign*}
\begin{flalign*}
    l_{i;\mu_q, \mu_{q^*}}=\int_{-\infty}^{\tau_0}A'_{i; \mu_q} A'_{i; \mu_{q^*}}g_i(t; (\bm \gamma_{\mu}^{\intercal}, \bm \gamma_{\sigma}^{\intercal})^{\intercal})dt, 
\end{flalign*}
\begin{flalign*}
    l_{i;\mu_q, \sigma_{q'}}=\int_{-\infty}^{\tau_0}A'_{i; \mu_q} A'_{i; \sigma_{q'}}g_i(t; (\bm \gamma_{\mu}^{\intercal}, \bm \gamma_{\sigma}^{\intercal})^{\intercal})dt
\end{flalign*}
and 
\begin{flalign*}
    l_{i;\sigma_{q'}, \sigma_{q^{+}}}=\int_{-\infty}^{\tau_0}A'_{i; \sigma_{q'}} A'_{i; \sigma_{q^{+}}}g_i(t; (\bm \gamma_{\mu}^{\intercal}, \bm \gamma_{\sigma}^{\intercal})^{\intercal})dt. 
\end{flalign*}
 
 \begin{theorem}
 The Fisher information matrix for $\bm \theta$ is given by $F(\bm \theta)=n \sum_{i=0}^m \pi_i l_i(\bm \theta)$ where 
 $$l_i(\bm \theta)=\begin{pmatrix}
 ((l_{i;\mu_q, \mu_{q^*}}))^{(Q_1+1) \times (Q_1+1)} & ((l_{i;\mu_q, \sigma_{q'}}))^{(Q_1+1) \times (Q_2+1)}\\
\left(((l_{i;\mu_{q}, \sigma_{q'}}))^{(Q_1+1) \times (Q_2+1)}\right)^{\intercal} & ((l_{i;\sigma_{q'}, \sigma_{q^{+}}}))^{(Q_2+1) \times (Q_2+1)}
\end{pmatrix},    
$$
$\pi_i$ is the proportion of allocation to the stress level $s_i$ and $\bm \theta=(\bm \gamma_{\mu}^{\intercal}, \bm \gamma_{\sigma}^{\intercal})^{\intercal}$.
\end{theorem}

\begin{proof}
The contribution of $L_i$ towards the Fisher information matrix $F(\bm \theta)$ for Type-I censored data is given by  
 \begin{equation}
    \int_{-\infty}^{\tau_0} \left(\frac{\partial \ln h_i(t; \bm \theta)}{\partial \bm \theta} \right) \left(\frac{\partial \ln h_i(t; \bm \theta)}{\partial \bm \theta} \right)^{\intercal}  g_i(t; \bm \theta) dt;
 \end{equation} see \cite{lawless2011statistical, gertsbakh1995fisher, park2009simple}. Using basic algebra on Equation (\ref{haz}), it can be shown that
 \begin{equation*}
 \ln  h_i(t; \bm \theta)=\frac{t-\mu(\xi_i;\bm \gamma_{\mu})}{\sigma(\xi_i; \bm \gamma_{\sigma})}-\ln \sigma(\xi_i; \bm \gamma_{\sigma}) =A(\bm \gamma_{\mu}, \bm \gamma_{\sigma}; t, \xi_i)=A_{i}.   
 \end{equation*}
To conclude the proof, note that  $l_{i;\mu_q, \mu_{q^*}}$, $l_{i;\mu_q, \sigma_{q'}}$ and $l_{i;\sigma_{q'}, \sigma_{q^{+}}}$ are obtained from $A_i=\ln h_i(t; \bm \theta)$ by taking partial double-derivatives with respect to $\gamma_{\mu, q}$ and $\gamma_{\sigma, q'}$, and then taking $F(\bm \theta)=n \sum_{i=0}^m \pi_i l_i(\bm \theta)$ (\citealp{chakrabarty2022economic}).  
\end{proof}

 \begin{result}
 The associated covariance matrix can be approximated by  $F^{-1}(\bm \theta)=(F(\bm \theta))^{-1}$ under mild regularization conditions. $F^{-1}(\bm \theta)$ is a matrix of order $(Q_1+Q_2+2) \times (Q_1+Q_2+2)$. Let us denote \begin{equation} \label{Finv}
F^{-1}(\bm \theta)=\begin{pmatrix}
 H^{11} & H^{12}\\
 H^{21} & H^{22}
\end{pmatrix},    
\end{equation} 
where $H^{11}, H^{12}=H^{21 \intercal}$ and $H^{22}$ are, respectively, the covariance submatrices corresponding to only $\bm \gamma_{\mu}$ (order $\overline{Q_1+1} \times \overline{Q_1+1}$), $\bm \gamma_{\mu}$ and $\bm \gamma_{\sigma}$ (order $\overline{Q_1+1} \times \overline{Q_2+1}$) and  only $\bm \gamma_{\sigma}$ (order $\overline{Q_2+1} \times \overline{Q_2+1}$). 
\end{result}

%%%%%%%%%%%%%%%%%%%%%%%%%%%%%%%%%%%%%%%%%%%%%%%%%%%%%%%%%%%%%%%%
%%%%%%%%%%%%%%%%%%%%%%%%%%%%%%%%%%%%%%%%%%%%%%%%%%%%%%%%%%%%%%%%
%%%%%%%%%%%%%%%%%%%%%%%%%%%%%%%%%%%%%%%%%%%%%%%%%%%%%%%%%%%%%%%%
%%%%%%%%%%%%%%%%%%%%%%%%%%%%%%%%%%%%%%%%%%%%%%%%%%%%%%%%%%%%%%%%
\section{Lot Acceptance Criterion and Asymptotic Distribution} \label{sec3}
A greater product lifetime is desirable when the quality characteristics are taken based on product lifetime. A lower specification limit ($l_s$) is often considered to ensure a basic or minimum level of product lifetime, hence, the product quality. The probability of an item being conforming or acceptable can be shown as $P(T_{ij}=\ln X_{ij} > \ln l_s)$. Let 
\begin{equation}\label{PLAdefmu0}
    \widehat{\mu_0}=\mu(\xi_0; \widehat{\bm \gamma}_{\mu})=\sum_{q=1}^{Q_1}\left[\widehat{\gamma}_{\mu, q}+ \left\{\frac{\widehat{\gamma}_{\mu, q}-\widehat{\gamma}_{\mu, q-1}}{\xi_{\mu,q}-\xi_{\mu, q-1}}\right\} (\xi_0-\xi_{\mu, q}) \right] I_{[\xi_{\mu,q-1},\xi{\mu,q}]}(\xi_0)
\end{equation} and 

\begin{equation}\label{PLAdefsig0}
    \widehat{\sigma_0}= \sigma(\xi_0; \widehat{\bm \gamma}_{\sigma})=\exp\left\{\sum_{q'=1}^{Q_2}\left[\widehat{\gamma}_{\sigma, q'}+ \left\{\frac{\widehat{\gamma}_{\sigma, q'}-\widehat{\gamma}_{\sigma, q'-1}}{\xi_{\sigma,q'}-\xi_{\sigma,q'-1}}\right\} (\xi_0-\xi_{\sigma,q'}) \right] I_{[\xi_{\sigma, q'-1},\xi_{\sigma, q'}]}(\xi_0)\right\}
\end{equation}
be the estimated values of product lifetime characteristics at the usage condition stress level $s_0$ or standardized stress level $\xi_0$ where $\widehat{.}$ indicates the ML estimate obtained by maximizing Equation (\ref{likeli}) with respect to $\bm \theta$. Let $W=\widehat{\mu_0} - k  \widehat{\sigma_0}$ with $k$ being the constant of acceptability. By \cite{lieberman1955sampling},
\begin{equation}
     W=\widehat{\mu_0} - k  \widehat{\sigma_0} > \ln l_s
\end{equation}
is considered the criterion based on which a lot of items can be accepted. Therefore, it is imperative to study the asymptotic distribution of $W$ to find out $k$ and to develop a sampling plan. The asymptotic distribution of $W$ is normal by the properties of the ML estimator.  

%%%%%%%%%%%%%%%%%%%%%%%%%%%%%%%%%%%%%%%%%%%%%%%%%%%%%%%%%%%%%%
%%%%%%%%%%%%%%%%%%%%%%%%%%%%%%%%%%%%%%%%%%%%%%%%%%%%%%%%%%%%%%
%%%%%%%%%%%%%%%%%%%%%%%%%%%%%%%%%%%%%%%%%%%%%%%%%%%%%%%%%%%%%%
%%%%%%%%%%%%%%%%%%%%%%%%%%%%%%%%%%%%%%%%%%%%%%%%%%%%%%%%%%%%%%
\subsection{Asymptotic Mean and Variance of $W$} \label{sec3_1}
To derive expressions for the asymptotic mean and variance of $W$, let us define the following terms for $q=0, \dots, Q_1$ and $q'=0, \dots,Q_2$: 
$$\xi'_{0,q}=\frac{\xi_0-\xi_{\mu, q}}{\xi_{\mu, q}-\xi_{\mu, q-1}}, \xi''_{0,q'}=\frac{\xi_0-\xi_{\sigma, q'}}{\xi_{\sigma, q'}-\xi_{\sigma, q'-1}}, \Omega({q-1},{q},q'-1,q')=[\xi_{\mu, q-1},\xi_{\mu, q}] \cap [\xi_{\sigma, q'-1},\xi_{\sigma, q'}],$$ 
$$g_{q,q'}(x_1, x_2, y_1, y_2)=\left(\left\{(1+\xi'_{0,q}) x_1 - \xi'_{0,q} x_2\right\} \exp\left\{(1+\xi''_{0,q'}) y_1 - \xi''_{0,q'} y_2\right\} \right),$$ for $(x_1, x_2, y_1, y_2)^{\intercal} \in \mathbb{R}^4$,
$\bm \gamma(q,q')=(\gamma_{\mu, q-1}, \gamma_{\mu, q}, \gamma_{\sigma, q'-1}, \gamma_{\sigma, q'})^{\intercal}$
and
\begin{equation*}
\tilde{H}(q,q')=
\begin{pmatrix}
 H^{11}_{q-1,q-1} & H^{11}_{q-1,q} & H^{12}_{q-1,q'-1} & H^{12}_{q-1,q'} \\
 H^{11}_{q-1,q} & H^{11}_{q,q} & H^{12}_{q,q'-1} & H^{12}_{q,q'} \\
 H^{12}_{q-1,q'-1} & H^{12}_{q,q'-1} & H^{22}_{q'-1,q'-1} & H^{22}_{q'-1,q'} \\
 H^{12}_{q-1,q'} & H^{12}_{q,q'}  & H^{22}_{q'-1,q'} & H^{22}_{q',q'} \\
\end{pmatrix}.    
\end{equation*}  
Let $f_N(.; \bm \gamma(q,q'), \tilde{H}(q,q'))$ be the pdf of four-dimensional multivariate normal distribution with mean vector $\bm \gamma(q,q')$ and covariance matrix $ \tilde{H}(q,q')$ and $\tilde{H}_{q,q'}$ denote $(q,q')$-th element of any matrix $\tilde{H}$.  

\begin{theorem}\label{varw} The asymptotic expectation and variance of $W=\widehat{\mu_0}-k\widehat{\sigma_0}$ are given by
\begin{flalign*}
    &\mathbb{E}(W)=\mu_0-k\sigma_0 \\ 
    &\mathbb{V}(W)=\sigma^2_{\mu_0}+k^2\sigma^2_{\sigma_0}-2k \sigma_{\mu_0, \sigma_0}
\end{flalign*}
where 
\begin{equation*}
    \sigma^2_{\mu_0}=\sum_{q=1}^{Q_1}\left[ H^{11}_{q,q} (1+\xi'_{0,q})^2 + H^{11}_{q-1,q-1} \xi^{'2}_{0,q} - 2 \xi'_{0,q} (1+\xi'_{0,q}) H^{11}_{q-1,q}\right] I_{[\xi_{\mu,q-1}, \xi{\mu,q}]}(\xi_0), 
\end{equation*}
\begin{equation*}
    \sigma^2_{\sigma_0}=\sigma_0^2\sum_{q'=1}^{Q_2}\left[ H^{22}_{q',q'} (1+\xi''_{0,q'})^2 + H^{22}_{q'-1,q'-1} \xi^{''2}_{0,q'} - 2 \xi''_{0,q'} (1+\xi''_{0,q'}) H^{22}_{q'-1,q'}\right] I_{[\xi_{\sigma,q'-1}, \xi_{\sigma,q'}]}(\xi_0), 
\end{equation*}
and 
\begin{flalign*}
    &\sigma_{\mu_0, \sigma_0}=\sum_{q=1}^{Q_1}\sum_{q'=1}^{Q_2}\left[ \iint\limits_{\mathbb{R}^4} g_{q,q'}(x_1, x_2, y_1, y_2) f_N((x_1,x_2,y_1,y_2)^{\intercal}; \bm \gamma(q,q'),\tilde{H}(q,q')) dx_1 dx_2 dy_1 dy_2 \right]\nonumber \\
    &\times I_{\Omega({q-1}, {q}, q'-1, q')}(\xi_0) - \mu_0\sigma_0.
\end{flalign*} 
\end{theorem}

\begin{proof}
First, note that $\widehat{\bm \theta}=(\widehat{\bm \gamma}_{\mu}^{\intercal}, \widehat{\bm \gamma}_{\sigma}^{\intercal})^{\intercal}$ asymptotically follows a multivariate normal distribution with mean vector $\bm \theta$ and covariance matrix $F^{-1}(\bm \theta)$ by the asymptotic properties of the ML estimator. This also means that any sub-vector of $\widehat {\bm \theta}$ follows multivariate normal (or normal if the sub-vector has only one element) distribution asymptotically with respective mean vector (or value) and covariance matrix (or variance) obtained from $F^{-1}(\bm \theta)$. \\

Second, Equations (\ref{PLAdefmu0}) and (\ref{PLAdefsig0}) can be re-written as 
\begin{equation*}
    \widehat{\mu_0}=\sum_{q=1}^{Q_1}\left[(1+\xi'_{0,q})\widehat{\gamma}_{\mu, q}  - \xi'_{0,q} \widehat{\gamma}_{\mu, q-1}\right] I_{[\xi_{\mu,q-1},\xi{\mu,q}]}(\xi_0)
\end{equation*} and 

\begin{equation*}
    \widehat{\sigma_0}=\sum_{q'=1}^{Q_2}\left[(1+\xi^{''}_{0,q'})\widehat{\gamma}_{\sigma, q'}  - \xi^{''}_{0,q'} \widehat{\gamma}_{\sigma, q'-1}\right] I_{[\xi_{\sigma,q'-1},\xi{\sigma,q'}]}(\xi_0).
\end{equation*}

Therefore, 
\begin{equation}\label{PLAmuexp}
    \mathbb{E}({\widehat{\mu_0}})=\left[(1+\xi'_{0,q})\mathbb{E} (\widehat{\gamma}_{\mu, q})  - \xi'_{0,q} \mathbb{E}(\widehat{\gamma}_{\mu, q-1})\right] I_{[\xi_{\mu,q-1},\xi{\mu,q}]}(\xi_0) = \mu_0, 
\end{equation}
\begin{equation}\label{PLAmuvar}
    \mathbb{V}({\widehat{\mu_0}})=\sum_{q=1}^{Q_1}\left[ H^{11}_{q,q} (1+\xi'_{0,q})^2 + H^{11}_{q-1,q-1} \xi^{'2}_{0,q} - 2 \xi'_{0,q} (1+\xi'_{0,q}) H^{11}_{q-1,q}\right] I_{[\xi_{\mu,q-1},\xi{\mu,q}]}(\xi_0), 
\end{equation}

\begin{equation}\label{PLAlnsigexp}
    \mathbb{E}({\widehat{\ln \sigma_0}})=\sum_{q'=1}^{Q_2}\left[(1+\xi^{''}_{0,q'})\mathbb{E}(\widehat{\gamma}_{\sigma, q'})  - \xi^{''}_{0,q'} \mathbb{E}(\widehat{\gamma}_{\sigma, q'-1})\right] I_{[\xi_{\sigma,q'-1},\xi{\sigma,q'}]}(\xi_0) = \ln \sigma_0 
\end{equation}
and 
\begin{equation}\label{PLAlnsigvar}
    \mathbb{V}(\widehat{\ln \sigma_0})=\sum_{q'=1}^{Q_2}\left[ H^{22}_{q',q'} (1+\xi''_{0,q'})^2 + H^{22}_{q'-1,q'-1} \xi^{''2}_{0,q'} - 2 \xi''_{0,q'} (1+\xi''_{0,q'}) H^{22}_{q'-1,q'}\right] I_{[\xi_{\sigma,q'-1},\xi{\sigma,q'}]}(\xi_0). 
\end{equation}
By the invariance property of the ML estimator and using Equations (\ref{PLAlnsigexp}) and (\ref{PLAlnsigvar}), we get $$\widehat{\sigma_0}= \exp\{\widehat{\ln \sigma}_0\}$$
as the ML estimator for $\exp\{\ln \sigma_0\}=\sigma_0$. Further, $\widehat{\sigma_0}$ follows a normal distribution asymptotically with mean $\sigma_0$ and variance $\mathbb{V}({\widehat{\sigma_0}})=\sigma_0^2 \mathbb{V}({\widehat{\ln \sigma_0}})$ by following the delta method; see \cite{casella2024statistical}. \\

Being the ML estimator, $(\widehat \gamma_{\mu, q-1},\widehat \gamma_{\mu, q},\widehat \gamma_{\sigma, q'-1},\widehat \gamma_{\sigma, q'})^{\intercal}$ follows a four-dimensional multivariate normal distribution with pdf $f_N(.; \bm \gamma(q,q'), \tilde{H}(q,q'))$. Thus, 
\begin{flalign}\label{expct}
    &\mathbb{E}(\widehat{\mu_0} \widehat{\sigma_0}) \nonumber\\
    &=\sum_{q=1}^{Q_1}\sum_{q'=1}^{Q_2} \mathbb{E}\left(\left\{(1+\xi'_{0,q}) \widehat{\gamma}_{\mu,q} - \xi'_{0,q} \widehat{\gamma}_{\mu,q-1}\right\} \exp\left\{(1+\xi''_{0,q'}) \widehat{\gamma}_{\sigma,q'} - \xi''_{0,q'} \widehat{\gamma}_{\sigma,q'-1}\right\} \right)\nonumber \\
    & \times I_{\Omega({q-1},{q},q'-1,q')}(\xi_0)\nonumber \\
    &=\sum_{q=1}^{Q_1}\sum_{q'=1}^{Q_2}\left[ \iint\limits_{\mathbb{R}^4} g_{q,q'}(x_1, x_2, y_1, y_2) f_N((x_1,x_2,y_1,y_2)^{\intercal}; \bm \gamma(q,q'),\tilde{H}(q,q')) dx_1 dx_2 dy_1 dy_2 \right]\nonumber \\
    &\times I_{\Omega({q-1},{q},q'-1,q')}(\xi_0).
\end{flalign}
Hence, from Equation (\ref{expct}), $\sigma_{\mu_0, \sigma_0}=\text{COV}(\widehat{\mu_0},\widehat{\sigma_0})=\mathbb{E}(\widehat{\mu_0} \widehat{\sigma_0})-\mu_0\sigma_0$ where COV(., .) denotes the covariance function. Therefore, this completes the proof if we consider $\sigma^2_{\mu_0}=\mathbb{V}(\widehat{\mu_0})$ and $\sigma^2_{\sigma_0}=\mathbb{V}(\widehat{\sigma_0})$.
\end{proof}
%%%%%%%%%%%%%%%%%%%%%%%%%%%%%%%%%%%%%%%%%%%%%%%%%
%%%%%%%%%%%%%%%%%%%%%%%%%%%%%%%%%%%%%%%%%%%%%%%%%
%%%%%%%%%%%%%%%%%%%%%%%%%%%%%%%%%%%%%%%%%%%%%%%%%
%%%%%%%%%%%%%%%%%%%%%%%%%%%%%%%%%%%%%%%%%%%%%%%%%

\subsection{Acceptability Constant and Optimization Constraint}\label{Sec3_2}
Let 
\begin{equation} \label{teststatstd}
    Z=\frac{W-(\mu_0 - k \sigma_0)}{\sqrt{\sigma^2_{\mu_0}+k^2\sigma^2_{\sigma_0}-2k \sigma_{\mu_0, \sigma_0}}}.
\end{equation}      
Therefore, $Z$ follows normal distribution asymptotically with mean $0$ and variance $1$. If $p_{nc}=P\left(T_{ij} \le {\ln l_s }\right)$ denotes the proportion of defectives or non-conforming items in a lot, then the OC function can be obtained as  
\begin{flalign}\label{Lpcurve}
L(p_{nc})&=P(W> \ln l_s|p_{nc})\nonumber \\
& = P\left(\frac{W-(\mu_0 - k \sigma_0)}{\sqrt{\sigma^2_{\mu_0}+k^2\sigma^2_{\sigma_0}-2k \sigma_{\mu_0, \sigma_0}}}>\frac{\ln l_s-(\mu_0-k \sigma_0)}{\sqrt{\sigma^2_{\mu_0}+k^2\sigma^2_{\sigma_0}-2k \sigma_{\mu_0, \sigma_0}}} \Bigg \vert p_{nc}\right) \nonumber \\
&=P\left(Z>\frac{\left[\frac{(\ln l_s-\mu_0)}{\sigma_0}\right] \sigma_0 + k\sigma_0}{\sqrt{\sigma^2_{\mu_0}+k^2\sigma^2_{\sigma_0}-2k \sigma_{\mu_0, \sigma_0}}}\Bigg \vert p_{nc}\right ) \nonumber \\
&=P\left(Z > \frac{({u_p}_{nc}+k)\sigma_0}{\sqrt{\sigma^2_{\mu_0}+k^2\sigma^2_{\sigma_0}-2k \sigma_{\mu_0, \sigma_0}}}\right)=1 - \Phi\left( \frac{({u_p}_{nc}+k)\sigma_0}{\sqrt{\sigma^2_{\mu_0}+k^2\sigma^2_{\sigma_0}-2k \sigma_{\mu_0, \sigma_0}}}\right)
\end{flalign} 
where ${u_p}_{nc}=\frac{\ln l_s-\mu_0}{\sigma_0}$ is the $p_{nc}$-th quantile of the standard extreme-valued (SEV) distribution where it has cdf of the form $G_s(t)=1-\exp\{-e^t\}, -\infty < t < \infty$. For $\alpha \in [0,1]$ and $\beta \in [0,1]$, let $p_{\alpha}$ and $p_{\beta}$ be the proportions of non-conforming or defective items in a lot that satisfy the producer's and consumer's risks, respectively. That is, the sampling plan must comply to the following inequalities:
\begin{equation} \label{prodrisk}
    P(W> \ln l_s| \text{proportion of non-conforming} \le p_{\alpha}) \ge 1 - \alpha
\end{equation}
and 
\begin{equation} \label{consrisk}
    P(W> \ln l_s| \text{proportion of non-conforming} \ge p_{\beta}) \le \beta.
\end{equation}
Using strict equations in the inequalities, (\ref{prodrisk}) and (\ref{consrisk}) yield 
\begin{equation}\label{zprod}
    \frac{(u_{p_{\alpha}}+k)\sigma_0}{\sqrt{\sigma^2_{\mu_0}+k^2\sigma^2_{\sigma_0}-2k \sigma_{\mu_0, \sigma_0}}}=z_{\alpha}
\end{equation}
and 
\begin{equation}\label{zcons}
    \frac{(u_{p_{\beta}}+k)\sigma_0}{\sqrt{\sigma^2_{\mu_0}+k^2\sigma^2_{\sigma_0}-2k \sigma_{\mu_0, \sigma_0}}}=z_{1-\beta}
\end{equation}
where $u_{p_{\alpha}}$ and $u_{p_{\beta}}$ are the $\alpha$-th and $\beta$-th quantiles of the SEV distribution, and $z_{\alpha}$ and $z_{1-\beta}$ are the $\alpha$-th and $(1-\beta)$-th quantiles of the standard normal distribution. Solving Equations (\ref{zprod}) and (\ref{zcons}), we obtain
\begin{equation}\label{c-const}
k=\frac{u_{p_{\alpha}}z_{1-\beta} - u_{p_{\beta}}z_{\alpha}}{z_{\alpha}-z_{1-\beta}}.    
\end{equation}
Fixing the points $(p_{\alpha}, 1- \alpha)$ and $(p_{\beta}, \beta)$ in the OC function, the following equality is obtained 
\begin{equation}\label{vcons}
 \frac{{\sigma^2_{\mu_0}+k^2\sigma^2_{\sigma_0}-2k \sigma_{\mu_0, \sigma_0}}}{\sigma_0^2} \left(\frac{z_{\alpha}-z_{1-\beta}}{u_{p_{\alpha}} - u_{p_{\beta}}}\right)^2=1.
\end{equation}
Equation (\ref{vcons}) confirms the risks that both the producers and consumers agree to tolerate and is treated as a constraint for the optimization problems discussed in the next section.
%%%%%%%%%%%%%%%%%%%%%%%%%%%%%%%%%%%%%%%%%%%%%%%%%%%%%%%%%%%%%%%%%%%%%%%
%%%%%%%%%%%%%%%%%%%%%%%%%%%%%%%%%%%%%%%%%%%%%%%%%%%%%%%%%%%%%%%%%%%%%%%
%%%%%%%%%%%%%%%%%%%%%%%%%%%%%%%%%%%%%%%%%%%%%%%%%%%%%%%%%%%%%%%%%%%%%%%
%%%%%%%%%%%%%%%%%%%%%%%%%%%%%%%%%%%%%%%%%%%%%%%%%%%%%%%%%%%%%%%%%%%%%%%
\section{Optimal Sampling Plan}\label{sec4}
The optimal sampling plans under accelerated life testing procedure where lifetime characteristics are modelled through piecewise linear functions are designed based on two approaches: aggregate cost minimization and variance of the aggregate quantile of log-lifetime minimization. Both approaches are elaborately discussed in several papers as mentioned in the Introduction. 

\subsection{Aggregate Cost Minimization} \label{sec4_1}
The first optimization problem is formulated to minimize the aggregated cost of the ALTSP. Following \cite{kwon1996bayesian}, acceptance, rejection, test duration, and inspection costs are taken into account to compute the aggregated cost. In particular, the acceptance cost is interpreted as the cost of the warranty. A general rebate warrant policy is assumed by \cite{kwon1996bayesian} following which an aggregate cost function is defined below as
\begin{flalign}\label{cmfun}
    C_T(n, \bm \pi, \bm \xi, \tau_0) &= (N-n) \left\{\sum_{i=0}^m \omega_i(\bm \theta)+\Phi\left( \frac{({u_p}_{nc}+k)\sigma_0}{\sqrt{\sigma^2_{\mu_0}+k^2\sigma^2_{\sigma_0}-2k \sigma_{\mu_0, \sigma_0}}}\right)\left(c_r-\sum_{i=0}^m \omega_i(\bm \theta)\right)\right\} \nonumber\\
    &+c_t \tau_0+ n c^{*}.
\end{flalign}
Here, $N$ is the lot or batch size, $n= n \sum_{i=0}^m \pi_i=\sum_{i=0}^m n_i$ is the number of items put to test, $\bm \pi=(\pi_0, \dots, \pi_m)^{\intercal}$, $\bm \xi=(\xi_0, \dots, \xi_m)^{\intercal}$, 
$$\omega_i(\bm \theta)=c_a\left\{\frac{\left[\omega_2 G_i(\omega_2)-\omega_1 G_i(\omega_1)\right]-\int_{\ln \omega_1}^{\ln \omega_2}e^t g_i(t)dt}{\omega_2-\omega_1}\right\}$$
with $G_i(.)$ and $g_i(.)$ defined in Equations (\ref{EQ2}) and (\ref{SEV_pdf}), respectively,
$c_r$ is the unit cost of rejection, $c_t$ is the per unit time consumption cost,  
%\begin{flalign*}
%\mathbb{E}(\tau)&=\mathbb{E}(\min\{T_{0j}, \tau_0\})\nonumber \\
%&=\tau_0 P(T_{ij}>\tau_0)+ \mathbb{E}(T_{ij}|T_{ij} \le \tau_0) P(T_{ij} \le \tau_0)\nonumber \\
%&=\tau_0 \overline{G}(\tau_o)+G(\tau_0)\int_{-\infty}^{\tau_0} \frac{tg(t)}{G(\tau_0)}dt %\nonumber \\ &=\tau_0 \overline{G}(\tau_o)+\int_{-\infty}^{\tau_0} tg(t)dt\nonumber \\
%&=\tau_0 \overline{G}(\tau_o)+\mu_0 \left[1- \exp\left\{-e^{\left(\frac{\tau_0-\mu_0}%{\sigma_0}\right)}\right\}\right] + \sigma_0\left[\psi^*\left(e^{\left(\frac{\tau_0-\mu_0}%{\sigma_0}\right)}\right)-\psi^*(0)\right]
%\end{flalign*}
$c^*$ is the unit cost of inspection and $\Phi(.)$ is the cdf of the standard normal distribution. %and $\psi^*(x)=\int_0^{x} (\ln{m}) e^{-m} dm$. 
Note that if the product under rebate warranty fails before $\omega_1$, free replacement is given with unit cost $c_a$, a pro-rata based cost is incurred if the product fails in $(\omega_1, \omega_2)$ and no cost of replacement is incurred if the product fails after $\omega_2$. Consequently, a relevant optimal design problem is given by 
\begin{flalign*}
   & \text{minimize } C_T(n, \bm \pi, \bm \xi, \tau_0) \nonumber \\
     & \text{subject to } \left[\frac{ {\sigma^2_{\mu_0}+k^2\sigma^2_{\sigma_0}-2k \sigma_{\mu_0, \sigma_0}}}{\sigma_0^2} \left(\frac{z_{\alpha}-z_{1-\beta}}{u_{p_{\alpha}} - u_{p_{\beta}}}\right)^2-1\right]=0, \sum_{i=0}^m \pi_i=1, \xi_0=0 \text{ and } \xi_m=1.
\end{flalign*}
The detailed algorithm for the constrained cost minimization problem is shown below.
\begin{algorithm}
\caption{Algorithm for aggregate cost minimization}\label{alg1}
%\par\noindent\rule{\textwidth}{0.4pt}
{\bf Given fixed quantities}: $\xi_0=0, \xi_m=1, \pi_0=0.20$, $\omega_1=0.50, \omega_2=0.75, c_a=0.15, c_r=0.80, c^*=0.05$ and $c_t=0.08$.\\

{\bf Input: }$N, \alpha, \beta, p_{\alpha}, p_{\beta}, \Xi_{\mu}=(\xi_{\mu, 0}, \dots, \xi_{\mu,Q_1})^{\intercal}, \bm \gamma_{\mu}=(\gamma_{\mu, 0}, \dots, \gamma_{\mu, Q_1})^{\intercal}$,  $\Xi_{\sigma}=(\xi_{\sigma, 0}, \dots, \xi_{\sigma,Q_2})^{\intercal}$ and $\bm \gamma_{\sigma}=(\gamma_{\sigma, 0}, \dots, \gamma_{\sigma, Q_2})^{\intercal}$.\\

{\bf Output: }$n, \xi_1, \dots, \xi_{m-1}, \pi_1, \dots, \pi_{m-1}, \tau_0, \ln{\tau_0}$ and $C_{\min}$.\\

{\bf Steps}:
\begin{itemize}
\item [1.] Provide initial values for $n, \xi_1, \dots, \xi_{m-1}, \pi_1, \dots, \pi_{m-1}$ and $\tau_0$. Note that $\bm \theta=(\bm \gamma_{\mu}^{\intercal}, \bm \gamma_{\sigma}^{\intercal})^{\intercal}$.
\item [2.] Following the steps leading to Equation (\ref{Finv}) and Theorem \ref{varw}, calculate $F^{-1}(\bm \theta)$  and $\sigma^2_{\mu_0}+k^2\sigma^2_{\sigma_0}-2k \sigma_{\mu_0, \sigma_0}$ respectively.
\item [3.] If $\sigma^2_{\mu_0}+k^2\sigma^2_{\sigma_0}-2k \sigma_{\mu_0, \sigma_0} \ge 0$, then define  the objective function $C_T(n, \bm \pi, \bm \xi, \tau_0)$ using Equation (\ref{cmfun}); else return ($10^{12}$). Here, $\bm \pi=(\pi_1, \dots, \pi_m)^{\intercal}$ and $\bm \xi=(\xi_1, \dots, \xi_{m-1})^{\intercal}$.
\item [4.] Define an inequality constraint sub-routine (INEQ) by specifying $0<\xi_1 < \dots < \xi_{m-1}<1$.
\item [5.] Further, Define an equality constraint sub-routine (EQLT) by specifying \\ $\left[\frac{\sigma^2_{\mu_0}+k^2\sigma^2_{\sigma_0}-2k \sigma_{\mu_0, \sigma_0}}{\sigma_0^2} \left(\frac{z_{\alpha}-z_{1-\beta}}{u_{p_{\alpha}} - u_{p_{\beta}}}\right)^2-1\right]=0$, $\sum_{i=0}^m \pi_i=1$, and $\xi_0=0 \text{ and } \xi_m=1$.    
\item [6.] Apply minimization sub-routine on $C_T(n, \bm \pi, \bm \xi, \tau_0)$ subject to INEQ and EQLT mentioned in steps 4 and 5, respectively.  {\tt nloptr} sub-routine within package {\tt NLOPTR} in {\tt R version 4.3.0} could be used. 
\item [7.] $C_{\min}$ (the minimized value of the objective function) along with corresponding parameters mentioned as output is obtained.
\item [8.] Calculate $n_i=\lfloor{n\pi_i}\rfloor$ for $i=1, \dots, m$ and $n_0=n-\sum_{i=1}^m n_i$.
\end{itemize}
%\par\noindent\rule{\textwidth}{0.4pt}
\end{algorithm}

\subsection{Variance of Aggregate Quantile of Log-Lifetime Minimization}\label{sec4_2}
The second optimization problem is formulated following \cite{kundu2008bayesian}. The details are as follows. Let ${T}_{0,b}$ is the $b$-th quantile of the EV log-lifetime distribution with location $\mu_0$ and scale $\sigma_0$ (true parameters at usage level) for $b \in [0,1]$. From Equation (\ref{EQ2}), $$1-\exp\left[-e^{\left(\frac{T_{0,b}-\mu_0}{\sigma_0}\right)}\right] = b,$$   
we obtain ${T}_{0,b}=\mu_0+\sigma_0 \ln\{-\ln b\}$. The ML estimate $\widehat{T_{0,b}}$ of ${T}_{0,b}$ is $\widehat{\mu_0}+\widehat{\sigma_0} \ln\{-\ln b\}$. Now, for $g^*(b)=\ln \{ - \ln b\}$, \begin{equation}\label{vq1}
\mathbb{V}(\widehat{T_{0,b}})=\sigma^2_{\mu_0}+{(g^*(b))^2}\sigma^2_{\sigma_0}+2 g^*(b)\sigma_{\mu_0, \sigma_0}
\end{equation}
where $\sigma^2_{\mu_0}$, $\sigma^2_{\sigma_0}$ and $\sigma_{\mu_0, \sigma_0}$ are defined in Theorem \ref{varw}.  Let the variance of the aggregate quantile of log-lifetimes is denoted by  $\mathcal{V}_Q(n, \bm \pi, \bm \xi, \tau_0)$, and is expressed by
\begin{flalign}\label{vq2}
    \mathcal{V}_Q(n, \bm \pi, \bm \xi, \tau_0)& =\int_{0}^1 \mathbb{V}( \widehat{T_{0,b}})db \nonumber \\ & = \sigma^2_{\mu_0}+\sigma^2_{\sigma_0} \int_0^{1}{(g^*(b))^2} db +2 \sigma_{\mu_0, \sigma_0} \int_0^{1} g^*(b)db \nonumber \\
    &=\sigma^2_{\mu_0}+\sigma^2_{\sigma_0} \left\{\gamma^2+\frac{\pi^2}{6}\right\} - 2 \gamma \sigma_{\mu_0, \sigma_0}
\end{flalign}
where $\gamma \approx 0.5772$ is the Euler-Mascheroni constant. Consequently, a relevant optimal design problem is given by 
\begin{flalign*}
   & \text{minimize } \mathcal{V}_Q(n, \bm \pi, \bm \xi, \tau_0) \\
     & \text{subject to } \left[\frac{ {\sigma^2_{\mu_0}+k^2\sigma^2_{\sigma_0}-2k \sigma_{\mu_0, \sigma_0}}}{\sigma_0^2} \left(\frac{z_{\alpha}-z_{1-\beta}}{u_{p_{\alpha}} - u_{p_{\beta}}}\right)^2-1\right]=0, \sum_{i=0}^m \pi_i=1, \xi_0=0 \text{ and } \xi_m=1.
\end{flalign*}
The algorithm for constrained variance minimization is detailed below.
\begin{algorithm}
\caption{Algorithm for the variance of the aggregate quantile of log-lifetime minimization}\label{alg2}
%\par\noindent\rule{\textwidth}{0.4pt}
{\bf Given fixed quantities}: $\xi_0=0, \xi_m=1, \pi_0=0.20$.\\

{\bf Input: }$N, \alpha, \beta, p_{\alpha}, p_{\beta}, \Xi_{\mu}=(\xi_{\mu, 0}, \dots, \xi_{\mu,Q_1})^{\intercal}, \bm \gamma_{\mu}=(\gamma_{\mu, 0}, \dots, \gamma_{\mu, Q_1})^{\intercal}$,  $\Xi_{\sigma}=(\xi_{\sigma, 0}, \dots, \xi_{\sigma,Q_2})^{\intercal}$ and $\bm \gamma_{\sigma}=(\gamma_{\sigma, 0}, \dots, \gamma_{\sigma, Q_2})^{\intercal}$.\\

{\bf Output: }$n, \xi_1, \dots, \xi_{m-1}, \pi_1, \dots, \pi_{m-1}, \tau_0, \ln{\tau_0}$ and $V_{\min}$.\\

{\bf Steps}:
\begin{itemize}
\item [1.] Provide initial values for $n, \xi_1, \dots, \xi_{m-1}, \pi_1, \dots, \pi_{m-1}$ and  $\tau_0$. Note that $\bm \theta=(\bm \gamma_{\mu}^{\intercal}, \bm \gamma_{\sigma}^{\intercal})^{\intercal}$.
\item [2.] Following the steps leading to Equation (\ref{Finv}) and Theorem \ref{varw}, calculate $F^{-1}(\bm \theta)$ and $\sigma^2_{\mu_0}+k^2\sigma^2_{\sigma_0}-2k \sigma_{\mu_0, \sigma_0}$ respectively.
\item [3.] If $\sigma^2_{\mu_0}+k^2\sigma^2_{\sigma_0}-2k \sigma_{\mu_0, \sigma_0} \ge 0$, then define  the objective function $\mathcal{V}_Q(n, \bm \pi, \bm \xi, \tau_0)$ using Equation (\ref{vq2}); else return ($10^{12}$). Here, $\bm \pi=(\pi_1, \dots, \pi_m)^{\intercal}$ and $\bm \xi=(\xi_1, \dots, \xi_{m-1})^{\intercal}$.
\item [4.] Define an inequality constraint sub-routine (INEQ) by specifying $0<\xi_1 < \dots < \xi_{m-1}<1$.
\item [5.] Further, define an equality constraint sub-routine (EQLT) by specifying \\  
 $\left[\frac{\sigma^2_{\mu_0}+k^2\sigma^2_{\sigma_0}-2k \sigma_{\mu_0, \sigma_0}}{\sigma_0^2} \left(\frac{z_{\alpha}-z_{1-\beta}}{u_{p_{\alpha}} - u_{p_{\beta}}}\right)^2-1\right]=0$, $\sum_{i=0}^m \pi_i=1$, and $\xi_0=0 \text{ and } \xi_m=1$.    
\item [6.] Apply minimization sub-routine on $\mathcal{V}_Q(n, \bm \pi, \bm \xi, \tau_0)$ subject to INEQ and EQLT mentioned in steps 4 and 5, respectively.  {\tt nloptr} sub-routine within package {\tt NLOPTR} in {\tt R version 4.3.0} could be used. 
\item [7.] $V_{\min}$ (the minimized value of the objective function) along with corresponding parameters mentioned as output is obtained.
\item [8.] Calculate $n_i=\lfloor{n\pi_i}\rfloor$ for $i=1, \dots, m$ and $n_0=n-\sum_{i=1}^m n_i$.
\end{itemize}
%\par\noindent\rule{\textwidth}{0.4pt}
\end{algorithm}

\section{Numerical Illustration} \label{sec5}
To demonstrate the optimal ALTSPs, a numerical study is conducted in this section with the help of algorithms 1 and 2 as discussed in Section \ref{sec4} and six choices of the input parameters $(\alpha, \beta, p_{\alpha}, p_{\beta})$ following MIL-STD-105D (\citealp{united1989sampling}). The choices in sequence are as follows: 
$$\text{ Case 1: } (0.050, 0.100, 0.021, 0.074),$$ 
$$\text{ Case 2: } (0.050, 0.100, 0.032, 0.094),$$  
$$\text{ Case 3: } (0.050, 0.100, 0.019, 0.054),$$ 
$$\text{ Case 4: } (0.100, 0.100, 0.021, 0.074),$$
$$\text{ Case 5: } (0.100, 0.100, 0.032, 0.094),$$
and
$$\text{ Case 6: } (0.100, 0.100, 0.019, 0.054).$$ 
The lot size $N$ is chosen as $1000$ for all cases. As discussed in Section \ref{sec4}, 5 stress levels are considered and the corresponding standardized stress levels are denoted by $\xi_0=0.000, \dots, \xi_4=1.000$. The sample size, i.e., the number of items to be tested under each stress level is taken as $n_i, i=0, \dots, 4$ with $n=\sum_{i=1}^4 n_i$. $\tau_0$ denotes the optimal censoring time to stop the ALT experiment. $C_{\min}$ and $V_{\min}$, respectively, represent the optimal minimum aggregate cost and minimum variance of the aggregate quantile of log lifetimes. The optimization results are provided in Table \ref{Table2} and Table \ref{Table3}. It can be observed from the tables that the minimum aggregate cost is obtained ($625.989$) for Case 2, while the minimum variance is obtained ($0.017$) for Case 5. Though an upper bound of 20\% of the lot size $N$ is set as the total sample size ($n$) for an experiment or plan, we note that optimal sample sizes are well below the upper limit for almost all cases. Hence, for example, if we observe Case 2 of Table \ref{Table2}, we may deduce that a total of $180$ items out of $1000$ items are put to test under ALT with 5 accelerated levels of standardized stress, namely, $0.000, 0.066, 0.323, 0.488$ and $1.000$ with respective sample sizes $37, 10, 54, 23$ and $56$. The ALT experiment has to be stopped at the censoring time of $5.551$ units. These optimal values lead to a minimum aggregate cost of $625.989$ units. \\

It can be observed that Case 2 results in the minimum aggregate cost whereas Case 5 results in the minimum overall variance of the quantiles. A natural skepticism may arise for deciding the set-up of input parameters that will yield the most optimized result. However, from a practitioner's point of view, it should depend on the intent and applicability. If one wants to minimize aggregate cost under general rebate warranty, then input parameter set-up corresponding to Case 2 can be chosen for the life testing experiment. On the other hand, if minimizing overall variance seems more plausible (as in the case where general rebate warranty does not apply or information is unavailable), then one may choose the input parameter set-up corresponding to Case 5. Nonetheless, the optimized overall variance of quantiles values for all cases are quite similar and low. Hence, even if one chooses Case 2 instead of Case 5, it seems unlikely to exhibit significantly higher overall variance of the quantiles.   

% Table generated by Excel2LaTeX from sheet 'CM Table'
\begin{table}[!htbp]
\centering
\caption{Results - Total cost minimization}\label{Table2}
%\resizebox{\columnwidth}{!}{
% Table generated by Excel2LaTeX from sheet 'VM Table'
\begin{tabular}{c|c|c|c|c|c|c}
\hline
Parameter & Case 1 & Case 2 & Case 3 & Case 4 & Case 5 & Case 6 \\
\hline
   $\alpha$ &   0.050 &      0.050 &      0.050 &      0.100 &      0.100 &      0.100 \\

   $\beta$ &  0.100 &      0.100 &      0.100 &      0.100 &      0.100 &      0.100 \\

    $p_{\alpha}$ &   0.021 &      0.032 &      0.019 &      0.021 &      0.032 &      0.019 \\

 $p_{\beta}$ &    0.074 &      0.094 &      0.054 &      0.074 &      0.094 &      0.054 \\

   $n$ &    199 &        180 &        134 &        140 &         95 &        197 \\

  $\xi_0$ &     0.000 &      0.000 &      0.000 &      0.000 &      0.000 &      0.000 \\

  $\xi_1$ &     0.087 &      0.066 &      0.068 &      0.052 &      0.054 &      0.060 \\

    $\xi_2$ &   0.319 &      0.323 &      0.350 &      0.320 &      0.321 &      0.345 \\

   $\xi_3$ &    0.476 &      0.488 &      0.580 &      0.414 &      0.614 &      0.604 \\

    $\xi_4$ &   1.000 &      1.000 &      1.000 &      1.000 &      1.000 &      1.000 \\

     $n_0$ &   42 &         37 &         29 &         30 &         21 &         41 \\

       $n_1$ & 20 &         10 &          8 &          7 &          4 &         17 \\

     $n_2$ &   54 &         54 &         41 &         20 &         28 &         68 \\

      $n_3$ &  47 &         23 &         13 &         17 &         17 &         13 \\

      $n_4$ &  36 &         56 &         43 &         66 &         25 &         58 \\

    $\ln \tau_0$ & 0.685 &      1.714 &      1.495 &      0.901 &      1.257 &      1.959 \\

    $\tau_0$ & 1.984 &      5.551 &      4.459 &      2.462 &      3.515 &      7.092 \\

   $C_{\min}$ & 650.768 &    625.989 &    688.780 &    658.168 &    645.163 &    652.406 \\
\hline
\end{tabular}  

%}\\
\footnotesize{}
\end{table}

% Table generated by Excel2LaTeX from sheet 'VM Table'
\begin{table}[!htbp]
\centering
\caption{Results - Aggregate quantile of log-lifetime variance minimization}\label{Table3}
%\resizebox{\columnwidth}{!}{
% Table generated by Excel2LaTeX from sheet 'VM Table'
\begin{tabular}{c|c|c|c|c|c|c}
\hline
Parameter & Case 1 & Case 2 & Case 3 & Case 4 & Case 5 & Case 6 \\
\hline
     $\alpha$ &      0.050 &      0.050 &      0.050 &      0.100 &      0.100 &      0.100 \\

      $\beta$ &      0.100 &      0.100 &      0.100 &      0.100 &      0.100 &      0.100 \\

   $p_{\alpha}$ &      0.021 &      0.032 &      0.019 &      0.021 &      0.032 &      0.019 \\

    $p_{\beta}$ &      0.074 &      0.094 &      0.054 &      0.074 &      0.094 &      0.054 \\

         $n$ &        157 &        189 &        149 &        164 &         41 &         45 \\

     $\xi_0$ &      0.000 &      0.000 &      0.000 &      0.000 &      0.000 &      0.000 \\

     $\xi_1$ &      0.054 &      0.054 &      0.207 &      0.099 &      0.125 &      0.121 \\

     $\xi_2$ &      0.338 &      0.315 &      0.362 &      0.321 &      0.304 &      0.346 \\

     $\xi_3$ &      0.485 &      0.613 &      0.555 &      0.507 &      0.458 &      0.415 \\

     $\xi_4$ &      1.000 &      1.000 &      1.000 &      1.000 &      1.000 &      1.000 \\

       $n_0$ &         33 &         39 &         31 &         34 &         11 &         10 \\

       $n_1$ &         31 &         12 &          9 &         20 &         10 &          4 \\

       $n_2$ &         32 &         37 &         37 &         25 &          5 &         13 \\

       $n_3$ &         24 &         57 &         12 &         62 &         13 &          4 \\

       $n_4$ &         37 &         44 &         60 &         23 &          2 &         14 \\

$\ln \tau_0$ &      0.941 &      1.930 &      1.657 &      1.338 &      1.218 &      1.525 \\

    $\tau_0$ &      2.563 &      6.890 &      5.244 &      3.811 &      3.380 &      4.595 \\

   V$_{\min}$ &      0.039 &      0.093 &      0.024 &      0.024 &      0.017 &      0.119 \\
\hline
\end{tabular}  
%}\\
\footnotesize{}
\end{table}

\section{Case Illustration and Comparison} \label{sec6}
One of the accomplishments of our developed model and method lies in capturing the non-linear dependence of product characteristics with respect to stress levels. Therefore, we consider a simulation study very similar to the one described in \cite{si2017accelerated} to demonstrate a better fit. An ALT experiment is assumed where product or component failure depends on external temperature (measured in Kelvin) applied to the system. The following well-known empirical relationships between true mean ($\mu^*$) and SD ($\sigma^*$) of product lifetimes (in hours) with a single stress factor temperature ($s$) are taken as: 

\begin{equation}\label{mu_s}
    \mu^*(s)=\frac{1}{R(s,v)} \left(\frac{v}{\delta_0}\right)^{\gamma_1} 
\end{equation}
and
\begin{equation}\label{sig_s}
    \sigma^*(s)=k_2 (\mu^*(s))^z 
\end{equation}
where $R(s, v)=\gamma_0 \exp\left\{-\frac{E_a}{k_1 s}+\gamma_2 \ln\{v\}+\frac{\gamma_3 \ln\{v\}}{k_1 s}\right\}$ and $v$ represents voltage (measured in Volt) applied to the component. Here, $\gamma_0=10^{-4},\gamma_1=1, \gamma_2=2, \gamma_3=1000, \delta_0=1$, $E_a=10000$ J/mol, ideal gas constant $k_1=8.314$ J mol$^{-1} K^{-1}$, $k_2=0.1$ and $z=1.2$ are experimental constants or component specific parameters (see \citealp{escobar2006review,si2017accelerated}). \\

For data simulation, we chose $v=170$V. The accelerating stress factor temperature $s$ is considered to have $m=6$, i.e., $7$ stress levels as $s_0=320$ (usage condition), $s_1=340, s_2=355, s_3=370, s_4=385, s_5=400$ and  $s_6=415$ K. The product lifetimes $x_{ij}$ are generated from Weibull distribution with shape parameter $\alpha_i=(\sigma^*(s_i))^{-1}$ and scale parameter $\lambda_i=\exp\{-\mu^*(s_i)\}$
for $i=0, \dots, m=6$ and $j=1, \dots, n_i=100$ for all $i$ with $\mu^*(s_i)$ and $\sigma^*(s_i)$  defined according to Equations (\ref{mu_s}) and (\ref{sig_s}). The pre-determined censoring time $\tau_0$ is chosen as $350$ hours which resulted in 15\% Type-I time censored data approximately. $t_{ij}=\ln\{x_{ij}\}$ are the log-lifetimes. Note that $\mu$ and $\sigma$ are the location and scale parameters of the EV distribution, and not the mean and SD. However, $\mu^*$ and $\sigma^*$ are the mean and SD of the simulated Weibull lifetimes.  \\

To the simulated data, two candidate models - a non-linear link-based PLA model as defined by Equations (\ref{PLAdefmu}) and (\ref{PLAdefsig}), and a linear link-based model are fitted. The stress levels $s_i$ are standardized first to obtain $\xi_i=\frac{s_i-s_{\min}}{s_{\max}-s_{\min}}$. The linear link based model assumes  $$\mu_i=\mu(\xi_i; \gamma_{\mu_0}, \gamma_{\mu_1})=\gamma_{\mu_0}+\gamma_{\mu_1}\xi_i$$ and $$\sigma_i=\sigma(\xi_i; \gamma_{\sigma_0}, \gamma_{\sigma_1})=\gamma_{\sigma_0}+\gamma_{\sigma_1}\xi_i$$ for $i=0, \dots, m=6$. The likelihood function for both models has the same structure as given in Equation (\ref{likeli}).  For the non-linear PLA model,  minimum, 33-rd percentile, 67-th percentile and maximum values, and minimum, 50-th percentile, and maximum values of the standardized stress levels are considered as the cut-points for linking $\mu$ and $\sigma$ to $\xi$ respectively. The log-likelihood function is maximized with respect to the parameters using sub-routine {\tt optim} with {\tt method="Nelder-Mead"}  in package {\tt optimx} in {\tt R version 4.3.0} for both models.\\

For the purpose of comparison, two criteria, namely, maximized log-likelihood value or $\hat l$ (i.e., value of the log-likelihood function at optimized parameter vector) and Akaike information criterion (AIC = $-2\hat l + \text{penalty}$) are considered. The `penalty' is the term that penalizes the likelihood function if there exists an overfitting (e.g., a large number of redundant parameters). For AIC, `penalty' is taken as $2$ times the number of parameters estimated during model fitting. The linear link-based model resulted in an average maximized log-likelihood value of $-7796.093$ (AIC$=15600.191$) with an empirical SD of $1096.484$ based on 100 replications. On the contrary, the non-linear link-based model resulted in an average maximized log-likelihood value of $-2390.911$ (AIC$=4795.823$) with an empirical SD of $328.419$ based on $100$ replication. We have chosen various numbers and values of the cut-points to link $\mu$ and $\sigma$ to $\xi$ in the case of the non-linear PLA-based model. However, since not much differences are noticed, the results for all choices are not reported for the sake of brevity. A stark difference in the average maximized log-likelihood value and AIC between the two models exists. This further establishes the requirement of a generalized link structure to capture the non-linearity present between lifetime characteristics and stress levels.  

\section{Concluding Remarks} \label{sec7}
Acceptance sampling plan plays a significant role in industrial manufacturing to reduce the time and cost of inspection of products. These plans require a more careful design when the lifetime of the product, in particular, comes with a warranty. The product lifetimes cannot be inspected conventionally. Appropriate censoring scheme(s) and accelerating factor(s) are required to carry out such inspection to observe failures in a limited time duration. Moreover, such a plan needs to ensure that the cost or any other technical criterion gets optimized. \\

Many researchers have shown their interest in developing such sampling plans, popularly known as accelerated life testing plans. The major assumption in these ALTSPs lies in the relationship between lifetime characteristics and accelerating stress factors which is assumed to be linear or log-linear in a majority of the cases. The very assumption may produce biased estimates which in turn may lead to an inefficient sampling design. This paper attempts to address the issue with a more flexible and general relationship. Consequently, the inclusion of non-linear dependence through the piecewise linear function is one of the major contributions in the paper to design an ALTSP.   \\

 Another contribution lies in the assumption of the non-identical nature of the Weibull distribution that is used to model product lifetime. The existing body of the literature in the ALTSP domain assumes the shape parameter of Weibull distribution as constant and the non-identical scale parameter is only shown to have a linear relationship with stress. The only article by \cite{seo2009design} assumed a non-identical Weibull distributed lifetime with both parameters linked with the stress factors linearly. This paper presents a design for ALTSP with Type-I censored non-identical Weibull distributed product lifetime with a non-linear relationship between the parameters (lifetime characteristics) and the stress factor. No such ALTSP, to the best of our knowledge, has been designed in the literature to date. The acceptability criterion for the conforming lots is formulated with the help of Fisher information matrix. The elements of the matrix are computed in great detail, and the asymptotic mean and variance of the acceptability criterion are also derived. \\ 

The proposed design minimizes the expected aggregated cost and the variance of the quantile of the log-lifetime distribution. The constraint of both optimization problems includes the mutually agreed upon producer's and consumer's risks. The entire design plan is demonstrated by a simulated case study. Under the same optimal criterion, it is found that the non-linear PLA-based relationship plan performs way better than the linear counterpart. The proposed ALTSP answers the following questions to meet the optimum criterion for given risks.
\begin{itemize}
\item [1.] How many products are to be tested?
\item [2.] What stress levels are to be used?
\item [3.] How many products are to be tested at each stress level? 
\item [4.] At what time will the experiment be terminated?
\end{itemize}

 The methodology used in the paper can be applied to any log location-scale family of distributions to design an optimal ALTSP. Although computationally intensive, designing an optimal plan using a spline with higher order can be interesting future work. Finally, other interesting problems for future research could include the design of similar ALTSP under the Bayesian framework. 

\section{Compliance with Ethical Standards}
{\bf Funding:} No funding has been received for this research.\\

{\bf Disclosure of potential Conflicts of Interest:} The authors report that there are no competing interests to declare. The authors have declared no conflict of interest.\\

{\bf Ethics approval:} This article does not contain studies with human participants or animals performed by any of the authors. \\

{\bf Informed Consent:} Not applicable.\\

{\bf Data Availability Statement:} No public or private data have been used in this investigation. Simulated data are used for demonstration. R codes are available upon request.\\

\bibliographystyle{apacite}
\bibliography{spline}

\end{document}